\documentclass{article}
\usepackage[utf8]{inputenc}
\usepackage{amsmath}
\usepackage{amsfonts}
\usepackage{amssymb}
\usepackage{mathtools}
\usepackage{xcolor,colortbl}
\usepackage{adjustbox}
\usepackage{amsthm}
\usepackage[toc,page]{appendix}
\usepackage{geometry}
\usepackage{bbm}
\newtheorem{theorem}{Theorem}
\newtheorem{corollary}{Corollary}
\newtheorem{lemma}{Lemma}

\title{Equilibria and Group Welfare in Vote Trading Systems}
\author{Matthew I. Jones}
\date{}

\begin{document}

\maketitle

\section*{Abstract}
We introduce a new framework to study the group dynamics and game-theoretic considerations when voters in a committee are allowed to trade votes. This model represents a significant step forward by considering vote-for-vote trades in a low-information environment where voters do not know the preferences of their trading partners. All voters draw their preference intensities on two issues from a common probability distribution and then consider offering to trade with an anonymous partner. The result is a strategic game between two voters that can be studied analytically. We compute the Nash equilibria for this game and derive several interesting results involving symmetry, group heterogeneity, and more. This framework allows us to determine that trades are typically detrimental to the welfare of the group as a whole, but there are exceptions. We also expand our model to allow all voters to trade votes and derive approximate results for this more general scenario. Finally, we emulate vote trading in real groups by forming simulated committees using real voter preference intensity data and computing the resulting equilibria and associated welfare gains or losses.

\section{Introduction}

For good reason, majority rule is the most widely-used decision rule for groups to consolidate members' preferences into a single selection, particularly for binary choices. This process is anonymous, decisive, and neutral~\cite{May_1952},
but it does not take into account voters' preference intensities; two voters that care deeply may be outvoted by three relatively ambivalent voters. While some may argue that this is by design~\cite{Bouton_Conconi_Pino_Zanardi_2021}, others conclude that sufficiently motivated minorities should be allowed to exert an oversized influence on certain issues~\cite{Jacobs_Christensen_Prislin_2009, Casella_2011}. Unfortunately, extracting preference intensities from voters is not straightforward, since they are incentivized to inflate the intensity of all their beliefs and claim that all issues are of paramount importance. When the group is deciding on multiple issues, one potential remedy for this shortcoming of majority rule is to allow voters to trade votes across issues, accumulating extra votes on their most valued issues and giving up their autonomy on issues they view as unimportant.

The study of vote trading has a long and complicated history that examines many different forms of voting and exchanging of votes to account for preference intensity (see \cite{Casella_Macé_2021} for a recent review). The central question, approached from many angles but not completely resolved, is ``Does the trading of votes improve outcomes for the entire group?'' The trading of votes-for-votes in a majority rule system introduces several complications that make votes unlike traditional goods in a market and require new analyses to make meaningful statements about the value of any particular trade.

Early theorists intuited that vote trading could address two issues in social choice: majority cycles and failure to respond to preference intensity~\cite{Buchanan_Tulock_1962, Coleman_1966, Mueller_1967}, but both claims are dubious at best. The tendency to remove majority cycles was quickly refuted by Park~\cite{Park_1967} and others~\cite{Miller_1977}. It took only slightly longer to show that although it is straightforward to create situations where the vote trading adds value to the group, it is equally simple to create scenarios where vote trading reduces value or even leads to Pareto inferior outcomes~\cite{Riker_Brams_1973, Ferejohn_1974}.

After an initial flurry of activity in the 1970s, the study of vote trading slowed down and diversified, studying various properties of vote trading and different mechanisms by which votes can be traded. One branch looks at vote trading as a dynamical system, studying stability and the convergence of trading to Condorcet winners~\cite{Casella_Palfrey_2019, Casella_Palfrey_2021}. There are also other ways in which votes can be exchanged besides the classic votes-for-votes framework. One line of research has examined ``implicit'' vote trading by combining multiple issues into a single bundle~\cite{Jones_Chervenak_Christakis_2023, Câmara_Eguia_2017}. Another way to trade votes is to exchange them for a numeraire, an alternative currency, that has value for voters and can be used to buy and sell votes~\cite{Casella_Llorente-Saguer_Palfrey_2012, Xefteris_Ziros_2017, Lalley_Weyl_2018}. Finally, a voter can trade votes with themselves by shifting votes from one issue to another~\cite{Casella_2011, Jackson_Sonnenschein_2007}. Each trading framework differs from vote-for-vote trading in important ways, and each requires its own analysis. 

This paper returns to the classic vote-for-vote model and attempts to study its effect on group welfare in a probabilistic framework. Samsonov, Sol\'{e}-Oll\'{e}, and Xefteris also recently looked at trading votes for votes, but in a system where vote count determines intensity of the reform, removing the payoff discontinuity at the pivotal vote and transforming the voting system back into a traditional market~\cite{Samsonov_Xefteris_SoleOlle_2023}.

In addressing the question of vote trading's impact on group welfare, one must first define group welfare. Early work, driven by economists, tended to focus on Pareto efficiency, with Ferejohn~\cite{Ferejohn_1974} going so far as to say ``Of course Mueller may be using some other notion of welfare [other than Pareto improvements] in which case it is not possible to decide the question of whether vote trading can produce welfare gains.'' 
Like many other papers to study vote trading, this paper uses a utility model where each voter gains utility when issues are accepted or rejected by the group~\cite{Riker_Brams_1973, Casella_Palfrey_2019}
and we define the welfare of the group to be the sum of the voters' final utilities. The benefit of this analysis is that the value of a hypothetical passionate minority is quantified and directly measured against the indifferent majority and that the value of the group decision on each issue can be examined in isolation, unlike traditional Pareto efficiency and Condorcet approaches.

Our model incorporates aspects from other papers as well. Like Xefteris and Ziros~\cite{Xefteris_Ziros_2017}, our voters have incomplete information, forcing them to make decisions in an ambiguous environment. Much of the literature around vote trading assumes that voters have complete information about the preferences (and preference intensities) of the entire voting population. This severely restricts the space of rational trades to only pivotal votes, where the trade actually changes the outcome of the vote. This is an assumption worth examining for two reasons. First, even in the most public and high-profile spaces, voters' preferences can be kept secret~\cite{Ramzy_2017, McPherson_2019}. Second, in a system where voters are required to state their preferences, they will be incentivized to lie about their preferences to make themselves more appealing as trading partners. Since preferences are often unknown (and when they are known, they may be falsified), we assume that voters have no beliefs about the preferences of others. All voters operate with the same knowledge about the distribution from which voter utilities are sampled but have no information about the preferences of their partners. In this paper, we analyze a model of vote trading that represents a critical divergence from the majority of vote trading work. Without complete information, voters must be concerned not just with the value of the issues they are trading on, but also if their trade will make any difference at all!  

In the main model, our voters are short-sighted, making myopic decisions as if no other voters will make trades~\cite{Casella_Palfrey_2018}. We loosen this assumption in Section \ref{sec:group_trades} and approximate the effect of randomly pairing voters and allowing them all to trade. Myopic trading also means that voters make sincere trades, working to earn votes that they think are valuable, instead of acquiring votes only to trade them away later~\cite{Iaryczower_Oliveros_2016}.

As the toy model presented in this paper shows, a voter's decision to offer a potential trade is dependent on her assessment of what her trading partner's preferences will be. Since partner behavior is critical in determining one's own behavior, we study this system via Nash equilibria in which no participants can improve their payoff by changing strategy~\cite{Holt_Roth_2004}. In this paper, we show how to find these equilibria and then prove several interesting properties of these vote trading systems. We also compute the value of vote trading for the group, specifically we find the probability that a trade has positive expected value for the net utility of the group, which can range from 0 to 1 depending on the underlying utility distribution. 

Finally, we use this new tool on real voter preference intensity data~\cite{ANES_2020}. Gathering empirical vote trading data is nearly impossible and most empirical vote trading studies are conducted in the laboratory with artificially-assigned utilities~\cite{Casella_Palfrey_2021, Tsakas_Xefteris_Ziros_2020, Goeree_Zhang_2017}. This analysis illustrates how this analysis of vote trading equilibria can be utilized alongside empirical data. In this example, we study the effect of vote trading in a hypothetical committee made up of real voters deciding on real issues, but this can be easily repeated on data for real decision-makers, assuming accurate information about such preferences could be collected.

This work contributes to the study of vote trading in several ways. First, it brings new mathematical techniques to bear on an old problem. This new perspective showcases the importance of the utility distribution of voters in determining the value of vote trading to both the voters and the entire group. Second, it provides a way of directly computing the probability that vote trading adds (or subtracts) value. While distributions exist where trading is extremely valuable, in most scenarios, vote trading removes value from the group by overriding simple majority rule. And third, it demonstrates a new way to study vote trading, by taking real voter preference data and predicting what trades would be offered by rational voters.

\section{A Model of Vote Trading}

We begin with a group of $n$ voters $v_i$, where $n$ is odd to avoid ties. Unless otherwise noted, we use $n=11$ in our examples. Since we consider one-for-one trades of votes, only the two issues are ever relevant when deciding to offer a trade and additional issues do not appear in our analysis, so for all intents and purposes, we only need to work with two issues, $t_1$ and $t_2$.

The vote trade game takes place in three stages. First, voters are assigned utilities on both issues which can take any value between -1 and 1, creating two vectors $T_j = \{u_{i,j} \in [-1,1] | i = 1 \dots n \}$ for $j = 1, 2$. Many, if not most, issues in real life are related, and there can be strong correlations between issues. We take this into account by drawing utilities $u_{i,1}, u_{i,2}$ jointly according to the probability distribution $f:[-1,1]^2 \to \mathbb R$ which can take any form desired. This function is known to all voters and is the factor that determines the equilibrium state(s). We assume throughout that $f$ is continuous and positive almost everywhere. We can integrate this function over the four half-squares to get the probabilities that a random voter has positive or negative value on issues $t_1$ and $t_2$, denoted $Q_1^+$, $Q_1^-$, $Q_2^+$, and $Q_2^-$.

In the second stage, two voters can offer to trade away their vote on $t_1$ for an additional vote on $t_2$, trade away $t_2$ for $t_1$, both, or neither. In the rare case that both voters offer to trade for either issue, the direction of the trade is chosen randomly.

Finally, all voters cast their votes on both issues. Voters are sincere, voting for an issue if $u_{i,j}>0$ and against it otherwise. However, if $v_1$ gave away their vote on $t_1$ to $v_2$, then $v_1$ votes for $t_1$ if $u_{2,1}>0$, and similarly, $v_2$ votes for $t_2$ if $u_{1,2}>0$. Let $\mathbb{I}_{t_j} = \begin{cases}
    1 & t_j \text{ passes} \\
    -1 & t_j \text{ fails}
\end{cases}$ represent the passage or failure of each issue. Then the final utility for each $v_i$ is $\mathbb{I}_{t_1} u_{i,1} + \mathbb{I}_{t_2} u_{i,2}$. When considering group welfare, we simply sum over all voters, which can be expressed in terms of the dot product as $\mathbb{I}_{t_1} \mathbbm{1} \cdot T_1 + \mathbb{I}_{t_2} \mathbbm{1} \cdot T_2$.

A voter's strategy is a decision of when to offer to trade away $t_1$ and/or $t_2$, and can be represented by two subsets of $[-1,1]$ as shown in Figure \ref{fig:dependent_both_players}. The calculation is different depending on the sign of $u_{i,1}$ and $u_{i,2}$, so we break these two subsets into the eight regions $R_1$ through $R_8$ which are determined by the angles $\theta_1$ through $\theta_8$. To get probabilities, we integrate $f$ over all these regions, so $I_i = \iint_{R_i} f(x,y) dx dy$.

\subsection{Nash Equilibria in Voting Systems}

Note that voters do not need to have opposing preferences on the issues when trading, so trades may not actually influence the final votes. Suppose we have a 5-voter system like in Table~\ref{tab:simple_low_info}. $v_1$ has utility $u_1=1$ on issue $t_1$ and $u_2=0.5$ on $t_2$. What could happen if $v_1$ offers to trade away her vote to $v_2$ on $t_2$ for $v_2$'s vote on $t_1$?

\begin{table}[h]
    \centering
    \begin{tabular}{r||c|c|}
         &  $t_1$ & $t_2$\\
         \hline \hline
        $v_1$ & 1 & 0.5 \\
        \hline
        $v_2$ & \cellcolor{gray!15}$u_4$ & \cellcolor{gray!15}$u_3$ \\
        \hline
        $v_3$ & \cellcolor{gray!15}? & \cellcolor{gray!15}? \\
        \hline        
        $v_4$ & \cellcolor{gray!15}? & \cellcolor{gray!15}? \\
        \hline
        $v_5$ & \cellcolor{gray!15}? & \cellcolor{gray!15}? \\
        \hline
    \end{tabular}
    \caption{A demonstration of the low-information environment. $v_1$ has no knowledge of the relative utilities (and therefore the voting behavior) of the other voters in the group, indicated by the grey shading.}
    \label{tab:simple_low_info}
\end{table}

If $u_4 < 0$, $v_1$ has successfully gained an additional vote on $t_1$ that may be the swing vote to change the final outcome, but if $u_4 > 0$, there will be no change in voting behavior since $v_2$ already supported $t_1$. On the other side, if $u_3 > 0$, $v_1$ doesn't have to change her vote on $t_2$ and doesn't have to risk being the swing vote that changes the outcome on the less-valuable issue, like she does if $u_3 < 0$. Of course, even if voting behavior changes, $v_1$ and $v_2$ can still be outvoted by $v_3$, $v_4$ and $v_5$. There are four events that are relevant to the expected value for the trader:

$A$: $v_1$ and $v_2$ have opposite preferences on $t_1$

$B$: $v_1$ and $v_2$ have opposite preferences on $t_2$

$C$: $v_2$ is the swing vote on $t_1$

$D$: $v_1$ is the swing vote on $t_2$

Either $A$ and $C$ need to happen, in which case value $2 |u_1|$ is gained, or $B$ and $D$ need to happen, in which case $2 |u_2|$ value is lost. Therefore, the expected value can be expressed as follows:

\begin{equation}\label{eq:expectation}
E(u_1,u_2) = 2 |u_1| P(A)P(C|A) - 2 |u_2| P(B)P(D|B)
\end{equation}

The myopic assumption is critical when computing the probability of $C$ and $D$. Voters make trades as if no other trades will take place; this allows us to treat the votes of the other $n-2$ voters as independent events that are only dependent on the underlying distribution $f$. We discuss how to account for the probability that a trade is made when computing the probability of pivotality in Section \ref{sec:group_trades}. Of course, the probabilities of these events depends on the trades that are being offered, so we need to compute the value of each trade and determine which trades should be offered all at the same time.

\begin{figure}
   \centering
   \includegraphics[width = \textwidth]{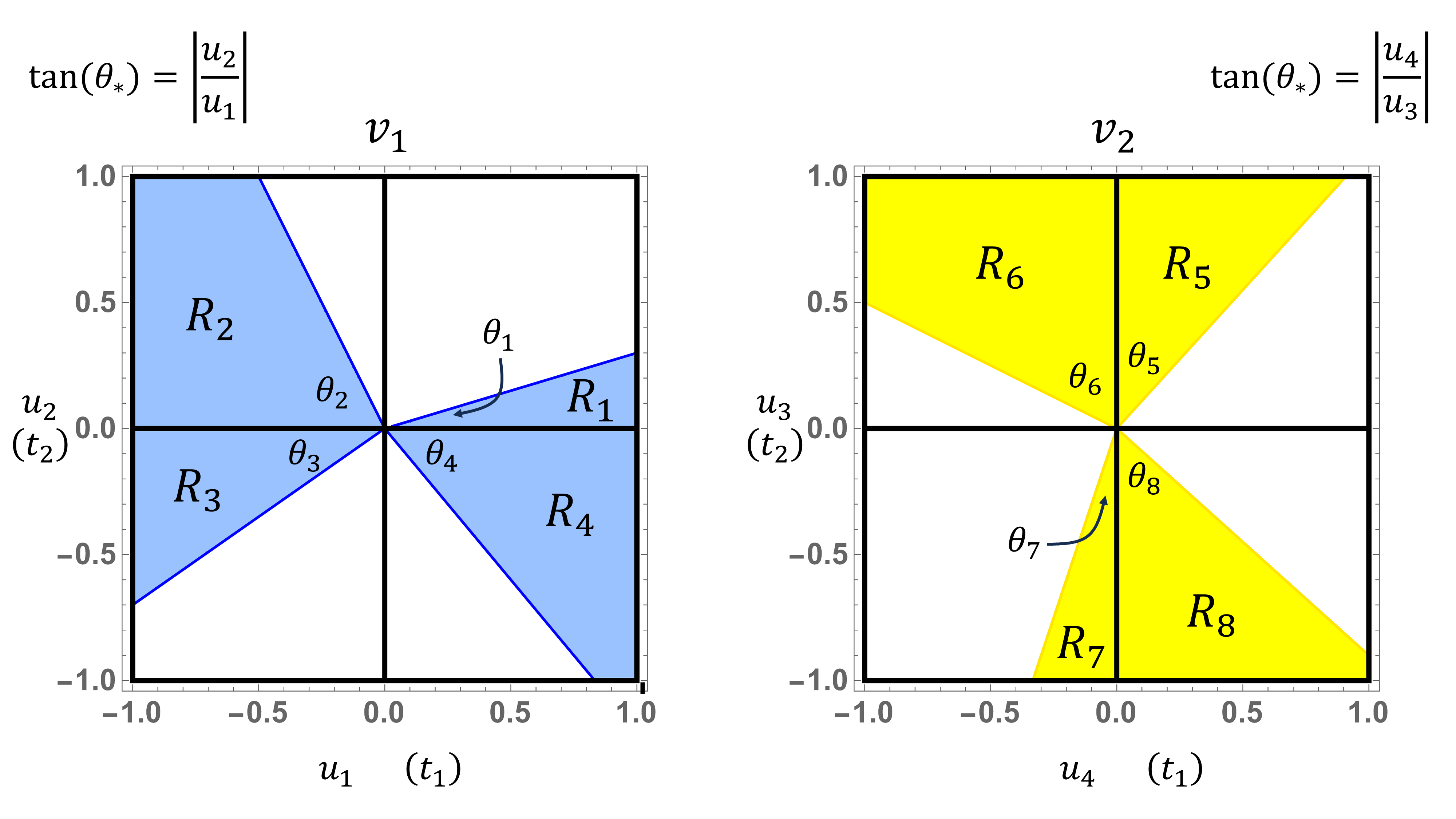}
   \caption{Trading regions for both players on two issues. For both plots, the utility on $t_1$ is represented by the x-axis, and the utility on $t_2$ is represented on the y-axis.  On the left, the trades that $v_1$ is offering highlighted in blue, while $v_2$'s trades are in yellow on the right. Each region $R_i$ is defined by the angle $\theta_i$ (alternatively, by the slope of the line). These angles will be adjusted to ensure that all trades in the blue or yellow regions have positive expected value and all trades in white have negative expected value}
   \label{fig:dependent_both_players}
\end{figure}

By finding the ratio of $\frac{u_2}{u_1}$ or $\frac{u_3}{u_4}$ where the value of the trade is zero (changing the slopes of the lines that bound the eight regions), we can separate the trades that have positive and negative value. We simultaneously solve this for all types of trades for $v_1$ and $v_2$ and the result is a Nash equilibrium.

\begin{theorem}\label{thm:ne}
Suppose $n$ voters are voting on issues $t_1$ and $t_2$ where utilities for the two issues are jointly distributed according to $f$. A set of trades being offered  (defined by the eight $\theta_i$ coefficients) is a non-trivial Nash equilibrium if and only if the $\theta_i$ satisfy the following equations:

\begin{equation}\label{eq:k1}
    \theta_1 = \arctan \left( \frac{I_6(\theta_6)+I_7(\theta_7)}{I_7(\theta_7)+I_8(\theta_8)}\frac{(Q_1^-)^\frac{n-1}{2} (Q_1^+)^\frac{n-3}{2}}{(Q_2^-)^\frac{n-3}{2} (Q_2^+)^\frac{n-1}{2}} \right)
\end{equation}

\begin{equation}\label{eq:k2}
    \theta_2 = \arctan \left( \frac{I_5(\theta_5)+I_8(\theta_8)}{I_7(\theta_7)+I_8(\theta_8)}\frac{(Q_1^-)^\frac{n-3}{2} (Q_1^+)^\frac{n-1}{2}}{(Q_2^-)^\frac{n-3}{2} (Q_2^+)^\frac{n-1}{2}} \right)
\end{equation}

\begin{equation}\label{eq:k3}
    \theta_3 = \arctan \left( \frac{I_5(\theta_5)+I_8(\theta_8)}{I_5(\theta_5)+I_6(\theta_6)} \frac{(Q_1^-)^\frac{n-3}{2} (Q_1^+)^\frac{n-1}{2}}{(Q_2^-)^\frac{n-1}{2} (Q_2^+)^\frac{n-3}{2}} \right)
\end{equation}

\begin{equation}\label{eq:k4}
    \theta_4 = \arctan \left( \frac{I_6(\theta_6)+I_7(\theta_7)}{I_5(\theta_5)+I_6(\theta_6)} \frac{(Q_1^-)^\frac{n-1}{2} (Q_1^+)^\frac{n-3}{2}}{(Q_2^-)^\frac{n-1}{2} (Q_2^+)^\frac{n-3}{2}} \right)
\end{equation}

\begin{equation}\label{eq:k5}
    \theta_5 = \arctan \left( \frac{I_3(\theta_3)+I_4(\theta_4)}{I_2(\theta_2)+I_3(\theta_3)} \frac{(Q_2^-)^\frac{n-1}{2} (Q_2^+)^\frac{n-3}{2}}{(Q_1^-)^\frac{n-3}{2} (Q_1^+)^\frac{n-1}{2}} \right)
\end{equation}

\begin{equation}\label{eq:k6}
    \theta_6 = \arctan \left( \frac{I_3(\theta_3)+I_4(\theta_4)}{I_1(\theta_1)+I_4(\theta_4)} \frac{(Q_2^-)^\frac{n-1}{2} (Q_2^+)^\frac{n-3}{2}}{(Q_1^-)^\frac{n-1}{2} (Q_1^+)^\frac{n-3}{2}} \right)
\end{equation}

\begin{equation}\label{eq:k7}
    \theta_7 = \arctan \left( \frac{I_1(\theta_1)+I_2(\theta_2)}{I_1(\theta_1)+I_4(\theta_4)} \frac{(Q_2^-)^\frac{n-3}{2} (Q_2^+)^\frac{n-1}{2}}{(Q_1^-)^\frac{n-1}{2} (Q_1^+)^\frac{n-3}{2}} \right)
\end{equation}

\begin{equation}\label{eq:k8}
    \theta_8 = \arctan \left( \frac{I_1(\theta_1)+I_2(\theta_2)}{I_2(\theta_2)+I_3(\theta_3)} \frac{(Q_2^-)^\frac{n-3}{2} (Q_2^+)^\frac{n-1}{2}}{(Q_1^-)^\frac{n-3}{2} (Q_1^+)^\frac{n-1}{2}} \right)
\end{equation}

\end{theorem}

\begin{proof}
See Appendix~\ref{sec:NE_deriv}.
\end{proof}

We describe a Nash equilibrium that satisfies Equations~\eqref{eq:k1}-\eqref{eq:k8} as non-trivial, in contrast to the trivial Nash equilibrium where no one offers any trades, and therefore no player can gain anything by offering to trade votes. This trivial equilibrium is not a strict Nash equilibrium, is present for all underlying joint utility distributions $f$, and is not very interesting from a voting perspective, so we restrict the rest of our discussion to non-trivial equilibria only.

Equations~\eqref{eq:k1}-\eqref{eq:k8} form a powerful tool to study equilibria in a vote trading system. For example, in other models, the existence of a stable outcome requires a great deal of effort~\cite{Casella_Palfrey_2019}, but here we get the existence of stability almost immediately.

\begin{corollary}\label{thm:existence}
    A non-trivial Nash equilibrium exists.
\end{corollary}

\begin{proof}
See Appendix~\ref{sec:exist_deriv}, in which we mimic the standard proof for the existence of a mixed Nash equilibrium~\cite{Nash_1951} using Brouwer's Fixed Point Theorem~\cite{Smart_1980} while bounding away from the origin and avoiding the trivial Nash equilibrium.
\end{proof}

\begin{figure}
    \centering
    \includegraphics[width=\textwidth]{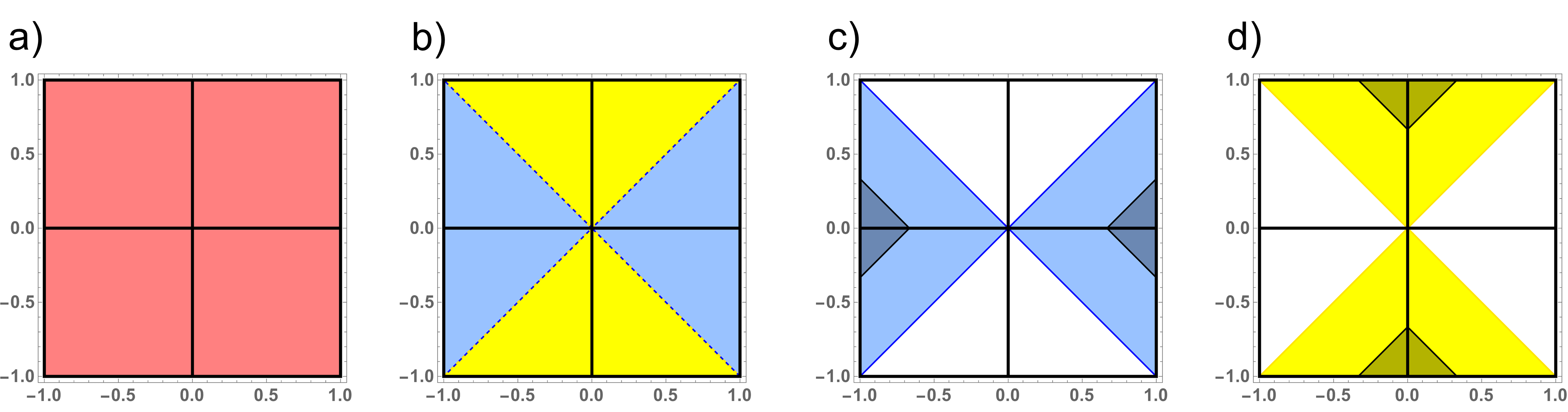}
    \caption{The results of our vote trading analysis on the uniform distribution. (a) shows a heatmap of the underlying distribution, in this case a constant $f(x,y) = 1/4$. (b) shows an equilibrium where the $\theta_i$ satisfy Equations~\eqref{eq:k1}-\eqref{eq:k8}. We call this particular equilibrium, where all $\theta_i = \frac{\pi}{4}$, the naive Nash equilibrium. (c) and (d) show the trades that improve group welfare in black alongside the trades that are offered by $v_1$ and $v_2$, respectively.}
    \label{fig:uniform_dist}
\end{figure}

Theorem~\ref{thm:ne} can also be used to study specific distributions of utilities. For example, in Figure~\ref{fig:uniform_dist}, we consider the uniform distribution $f(x,y)= 1/4$. Plugging in $f$ and solving all eight equations yields the solution $\theta_i = \frac{\pi}{4}$, illustrated in Figure~\ref{fig:uniform_dist}b. We refer to this equilibrium as the ``naive equilibrium'' in which all voters should offer to give up their vote on an issue that matters less in exchange for another vote on an issue that matters more. This leads us to a powerful symmetry result.

\begin{corollary}
    Suppose $f$ is point symmetric around the origin, meaning $f(x,y)=f(-x,-y)$. Then the naive state is, in fact, a Nash equilibrium.
\end{corollary}

\begin{proof}
    Suppose $\theta_i = \frac{\pi}{4}$ for all $i$. Since $f$ is point symmetric around the origin, $I_1(\frac{\pi}{4})=I_3(\frac{\pi}{4})$, $I_2(\frac{\pi}{4})=I_4(\frac{\pi}{4})$, $I_5(\frac{\pi}{4})=I_7(\frac{\pi}{4})$, and $I_6(\frac{\pi}{4})=I_8(\frac{\pi}{4})$. These equations make it straightforward to show that $Q_1^+ = Q_1^- = Q_2^+ = Q_2^-$ and then Equations~\eqref{eq:k1}-\eqref{eq:k8} follow easily.
\end{proof}

Even if $f$ is not point symmetric, we can still compute the Nash equilibrium, numerically if necessary. Suppose the two issues have utilities distributed according to the following joint distribution function:

\begin{equation}\label{eq:pw_const_dist}
    f(x,y) = \begin{cases}
        1/10 & x,y>0 \\
        2/10 & x>0, y<0 \\
        3/10 & x,y < 0 \\
        4/10 & x<0, y>0
    \end{cases}
\end{equation}

The equilibrium for this distribution is shown in Figure \ref{fig:pw_const_dist}. This case differs from the naive equilibrium because of the white and green regions in Figure \ref{fig:pw_const_dist}b. In the white regions, it is not worth trading either issue for the other, even though $t_1$ is more valuable than $t_2$. In the green regions, it is worth trading either issue for the other. 

\begin{figure}
   \centering
   \includegraphics[width = \textwidth]{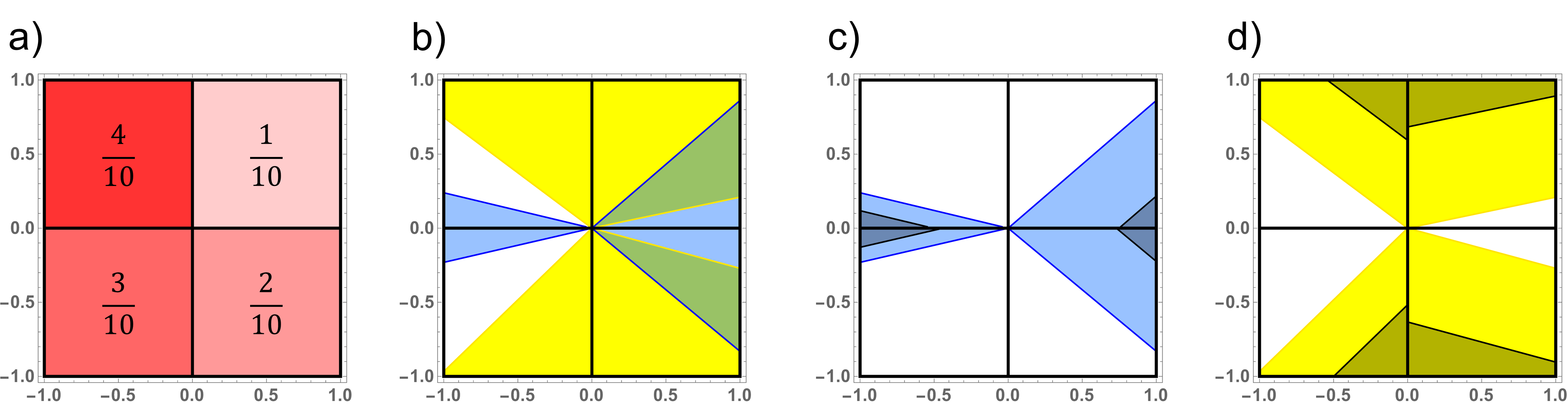}
   \caption{The results of our vote trading analysis when issues are dependent. (a) shows a heatmap of the underlying distribution, Equation \eqref{eq:pw_const_dist}. (b) shows the equilibrium. Trades in the green regions are profitable for $v_1$ and $v_2$, while positions in the white regions trades are not worth trading for either player. (c) and (d) show the trades that improve group welfare in black alongside the trades that are offered by $v_1$ and $v_2$, respectively.}
   \label{fig:pw_const_dist}
\end{figure}

For example, suppose a player has utility 0.8 on $t_1$ and -0.1 on $t_2$. Because $t_1$ is so much more important, this player would be willing to play the role of Player 1, giving away their vote on $t_2$ for an additional vote on $t_1$. This player should not consider giving away the vote on $t_1$ for a vote on $t_2$. This is what we would naively expect. 

Now suppose that the player has utility 0.8 and 0.4 on $t_1$ and $t_2$, respectively. It turns out that this player can expect a positive gain from trading away their vote on $t_1$ or $t_2$! Conversely, suppose the utilities are -0.9 and 0.4. This player would not expect to gain value by giving away either vote, and is content to have one vote on each issue. 

We present one final claim regarding equilibria of vote trading systems: 

\begin{corollary}
    There exist probability distributions $f$ with multiple solutions to Equations~\eqref{eq:k1}-\eqref{eq:k8} and therefore have multiple Nash equilibria.
\end{corollary}

\begin{proof}
    We prove this with a family of examples. Let $f$ be any joint probability distribution that is positive almost everywhere whose integrals in each of the regions shown is given by the values shown in Figure \ref{fig:nonunique_integrals}.

    Straightforward computations who that any such function will have at least two equilibria defined by the $m_i$ and $n_i$ coefficients. Therefore, our Nash equilibria are not necessarily unique. Notice that in Figure~\ref{fig:nonunique_integrals}, $Q_1^+ = Q_1^-=Q_2^+=Q_2^-=\frac{1}{2}$, so this example holds for all $n$. 
\end{proof}

\begin{figure}
   \centering
   \includegraphics[width=\textwidth]{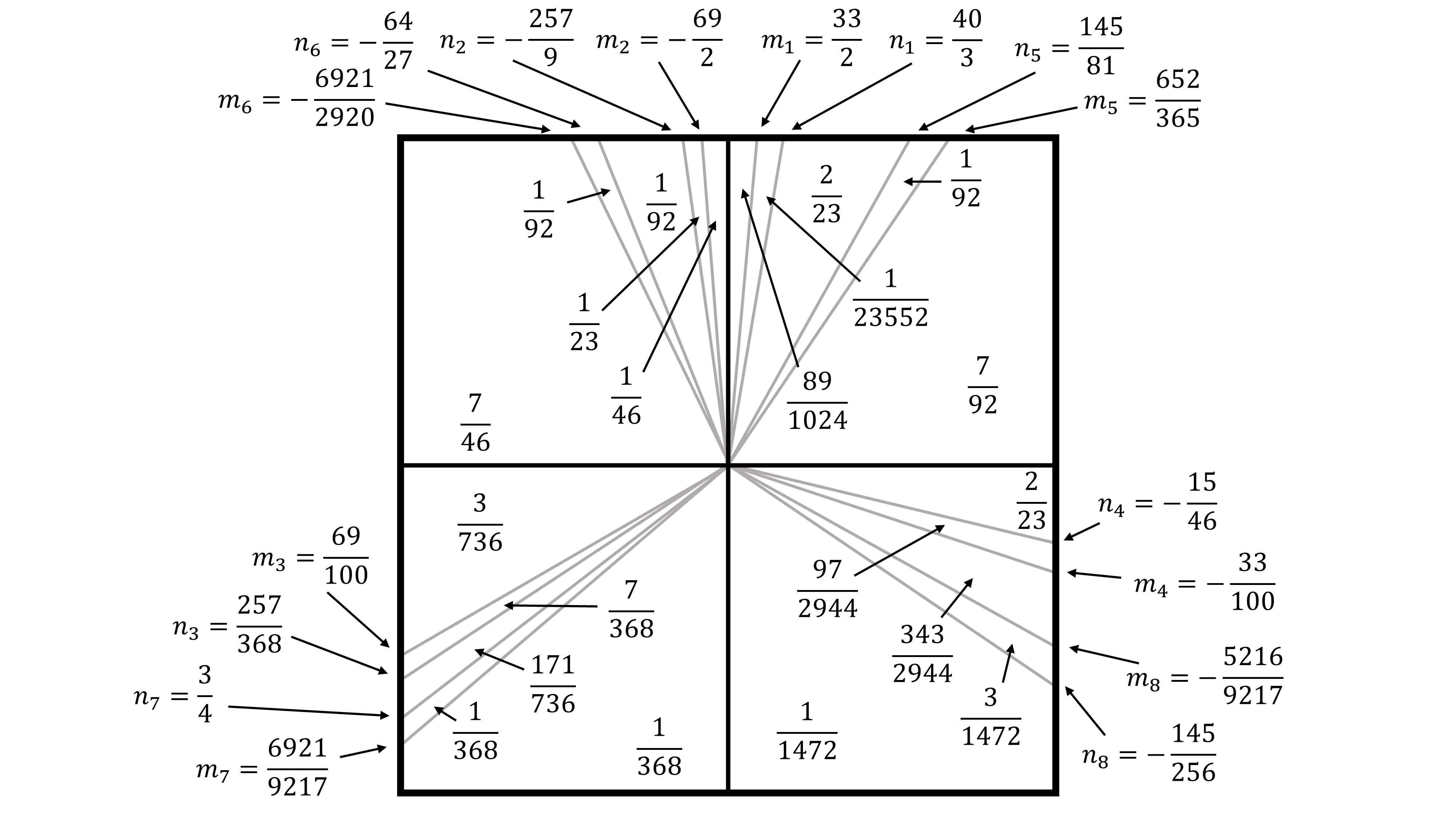}
   \caption{An integral representation of a function with two Nash equilibria. Each region is defined by one or two lines with slopes indicated by the $m_i$ and $n_i$. The integral of each region is indicated by the values inside the square. Note that for clarity, these regions are not necessarily to scale. Several of these regions are too small be represented accurately and be readable, so the drawn lines do not represent the true slopes.}
   \label{fig:nonunique_integrals}
\end{figure}

As a comment, the distribution shown in Figure \ref{fig:nonunique_integrals} has extremely high density variance. Some tiny slivers of the square have a large fraction of the total utility pairs (like the triangle defined by $n_2$ and $m_2$) while other large regions have essentially no mass (like the trapezoid bordered by $n_8$). All the counterexamples that we have generated seem to have this property, and we conjecture that all ``reasonable'' utility distributions have a unique non-trivial equilibria, although this seems difficult to state rigorously, let alone prove.

\subsection{Group Welfare}

Theorem~\ref{thm:ne} has given us a way to determine which sets of trades form an Nash equilibrium under any particular density function, so now we can consider the effect such trades will have on the group as a whole. Recall that we define the payoff from any decision for a group to be the sum of payoffs for the individuals in the group. This allows us to compute the expected value of a trade \emph{for the whole group} by extending Equation \eqref{eq:expectation} to account for the utilities of all members of the group. The derivations and resulting equations are uncomplicated but tedious, so we have placed them in Appendix~\ref{sec:group_welfare_deriv} and simply demonstrate visually with some examples in the main text by plotting the regions where trades have positive expected value for the group in black.

Again, we first turn to the uniform distribution, where we know that the naive equilibrium holds. Every pair of utilities is offered to trade, but only a fraction of utility pairs have a difference between utilities that is great enough to be worth overriding the majority decision. The black regions shown in Figure~\ref{fig:uniform_dist} (c) and (d) indicate utility pairs that are beneficial for the entire group. By integrating these regions (or rather, the intersection of these regions with the $R_i$ regions), we compute the probability that a trade improves group welfare. For the uniform distribution, this probability is exactly $\frac{1}{9}$. 

For the utility distribution given in Equation~\eqref{eq:pw_const_dist}, trades are actually slightly more valuable. By integrating the black regions of Figure~\ref{fig:pw_const_dist}, we see that about 18.5\% of trades have positive expected value for the entire group.

However, these examples do not imply that vote trading will always have a low probability of improving welfare. Consider the equation
\begin{equation}\label{eq:beneficial_dist}
g_\alpha(z) = \begin{cases}
-1^\alpha \frac{\alpha + 1}{2} z^\alpha & z < 0 \\
-1^\alpha \frac{\alpha + 1}{2} (z-1)^\alpha & z>0
\end{cases}\end{equation}
and let $f(x,y) = g_\alpha(x)g_\alpha(y)$. As $\alpha$ approaches $\infty$, the distribution becomes more and more skewed and the probability of a trade being beneficial approaches 1 since most types of trades become beneficial for the group and the few that do not become increasingly unlikely. In Figure \ref{fig:beneficial_dist} where $\alpha = 4$, the probability of a beneficial trade is around 95\%.

\begin{figure}
    \centering
    \includegraphics[width = \textwidth]{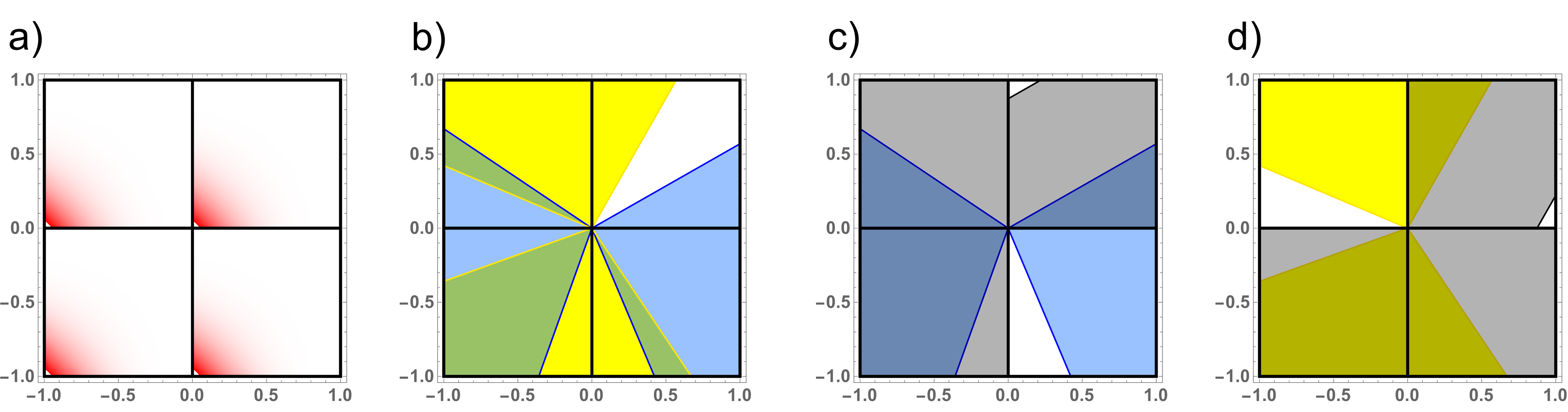}
    \caption{The results of our vote trading analysis on a highly skewed independent distribution. (a) shows a heatmap of the underlying distribution from Equation \eqref{eq:beneficial_dist} with $\alpha=4$. (b) shows the Nash equilibrium, in which many trades are offered. (c) and (d) show the trades that improve group welfare in black alongside the trades that are offered by $v_1$ and $v_2$, respectively. The regions of beneficial trades completely cover many of the regions in which trades are being offered, so the vast majority of trades are beneficial for the group.}
    \label{fig:beneficial_dist}
\end{figure}

We can also use this technique on simple distributions to draw generalized conclusions about the value of vote trading. For instance, consider two distributions, both symmetric with independent utilities, but where one has mainly uninterested voters and one with predominantly passionate voters (Figure \ref{fig:symmetric_dists} (a) and (d), respectively). Being point symmetric, the naive equilibrium (Figure \ref{fig:uniform_dist}b) holds for both cases, but trades have a very different impact on group welfare; the former has beneficial trades approximately 24.5\% of the time and the latter never has beneficial trades. 

This result is surprising in the context of previous research which suggests that vote trading tends to improve welfare when preferences are heterogeneous. The key distinction is that these previous works have assumed trading with a numeraire~\cite{Casella_Llorente-Saguer_Palfrey_2012} or in a power-sharing system~\cite{Tsakas_Xefteris_Ziros_2020}, where voters who give up a vote and lose can still gain a small degree of value or representation. In this model, the majoritarian system offers no such conciliation prize. Trading improves welfare only when the stakes are low enough that most voters are relatively indifferent to the outcome of the vote.

\begin{figure}
    \centering
    \includegraphics[width = \textwidth]{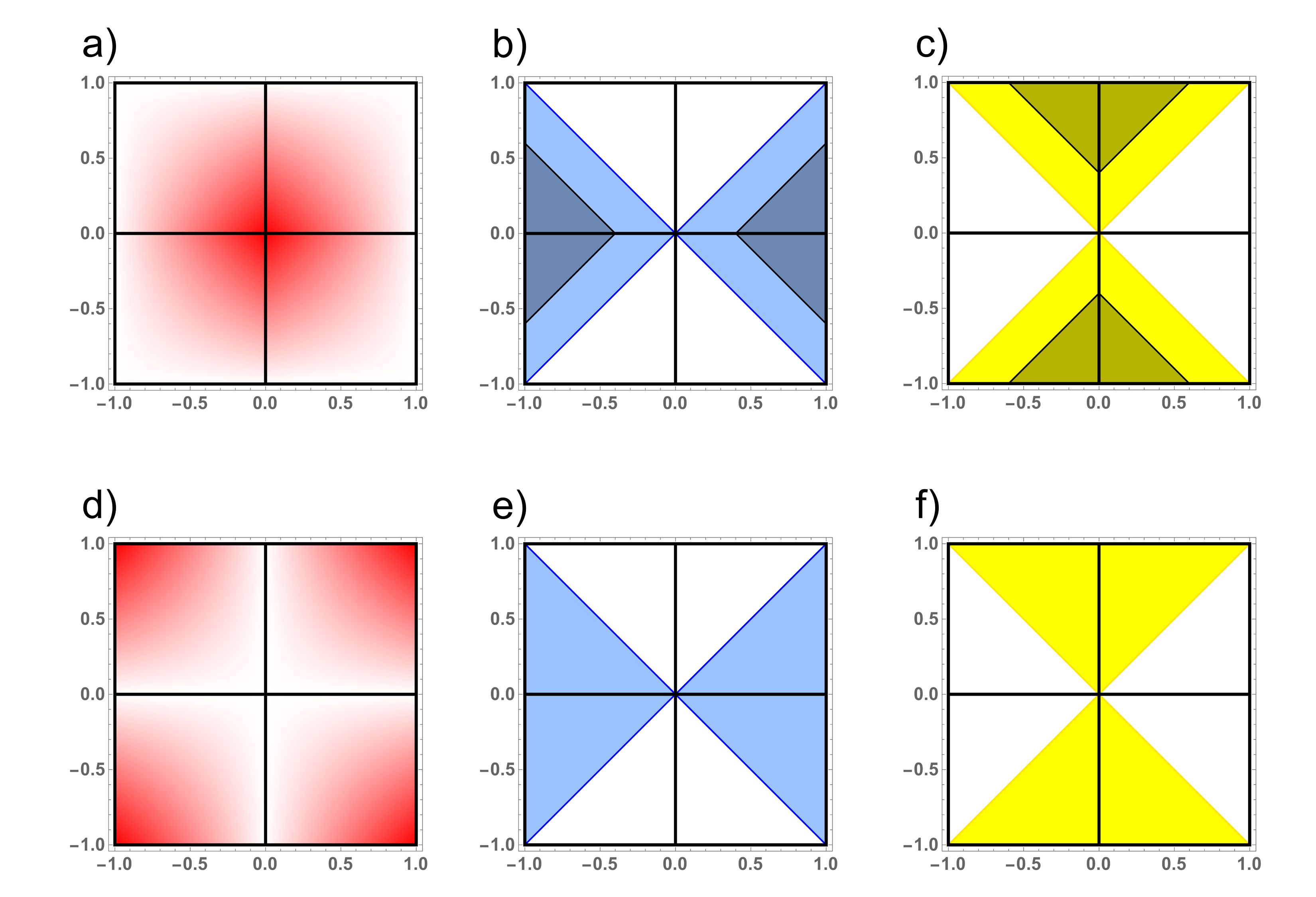}
    \caption{The difference in benefit of trades in two symmetric distributions. Both distributions have the form $f(x,y) = g(x)g(y)$ where $g$ is a symmetric 1D distribution. Subfigures (a), (b), and (c) correspond to $g(z) = \begin{cases}
        1+z & z<0 \\
        1-z & z\geq 0
    \end{cases}$
    while (d), (e), and (f) correspond to $g(z) = \begin{cases}
        -z & z<0 \\ 
        z & z\geq 0
    \end{cases}$. (a) and (c) show the distribution heatmaps. A group with many low utilities like (a) benefits from vote trading much more than a group with many extreme utilities like (d).}
    \label{fig:symmetric_dists}
\end{figure}

\section{Vote Trading in Real Populations}

These toy examples are good for generalized principles, but still a step away from determining if vote trading improves group welfare in the real world. It is very difficult to study vote trading empirically. In the political sphere, most lawmakers are unwilling to admit to voting against their preferences on an issue but the common suspicion is that vote trading is prevalent throughout legislatures~\cite{Stratmann_1995, Aksoy_2012, Cohen_Malloy_2014}. In this final section, we form a hypothetical committee made up of real voters with real preferences on real issues and examine the effect vote trading would have on this group.

The American National Election Survey~\cite{ANES_2020} regularly polls voters on a wide range of issues. On some issues, they ask for voters' preferences and the intensity of preferences, so this data allows us to estimate the joint probability distribution of American voter utilities on pairs of issues. The issues we consider here for illustrative purposes are government action on addressing economic inequality and the overall level of government regulation. These issues are highly correlated. For both issues, voters were asked about their preferences on a scale from 1 (strongly oppose) to 7 (strongly support). Using this discrete data, we apply Gaussian Kernel Density Estimation~\cite{Weglarczyk_2018} to get a continuous approximation for the joint probability distribution, which is shown in Figure~\ref{fig:anes_fig}a.

With this density function, we can find the Nash equilibrium and determine which trades are valuable. About 38.5\% trades are beneficial to the entire group. Thus, two voters trading votes on government action toward inequality and government regulation is probably harming group welfare by subverting majority rule, rather than benefiting the group by more accurately expressing voter preferences.

\begin{figure}
    \centering
    \includegraphics[width = \textwidth]{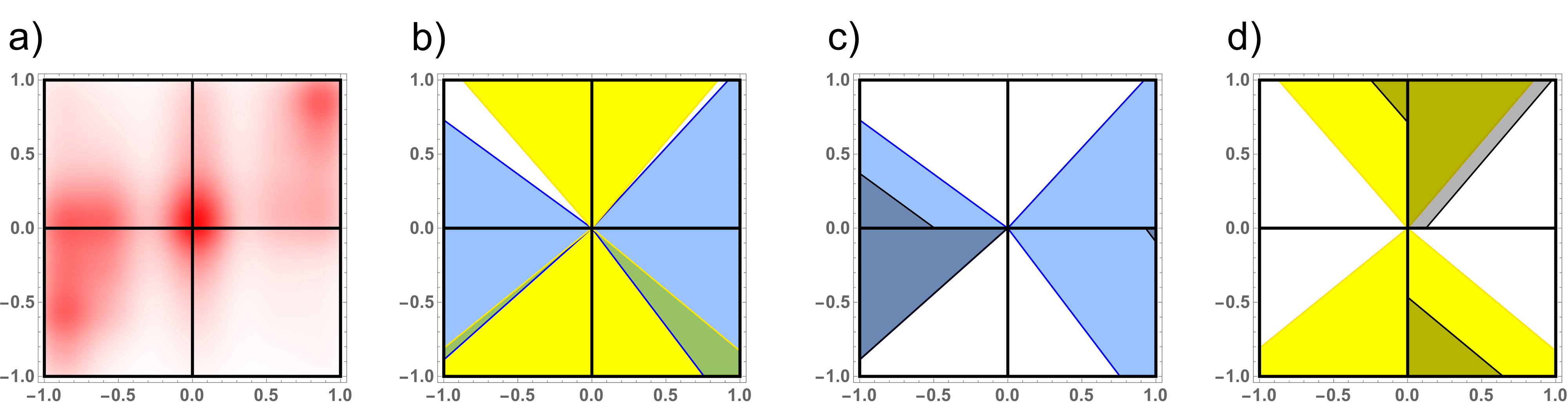}
    \caption{The results of our vote trading analysis on real voter data. (a) shows a heatmap of the underlying distribution, where the x-axis is utility on government action on inequality and the y-axis is utility on government regulation. (b) shows the equilibrium, which is close to the naive equilibrium, suggesting that the intuitive naive state can be a good heuristic for determining rational trading behavior. (c) and (d) show the trades that improve group welfare in black alongside the trades that are offered by $v_1$ and $v_2$, respectively.}
    \label{fig:anes_fig}
\end{figure}

\section{Group-wide Trading}\label{sec:group_trades}

Until now, voters assume that their trade is the only one that will be made, and that the other $n-2$ voters will vote sincerely. Of course, in most scenarios, all voters will have equal opportunity to trade votes, so in this section we extend our model to allow all voters the chance to trade. 

Like the single-trade model, we first assign all voters utilities on $t_1$ and $t_2$. Then, voters are randomly paired together and allowed to trade like before. Because $n$ is odd, one voter will be left out but the rest of the voters must take into account the other $\frac{n-3}{2}$ pairs of voters that may be trading. Unfortunately, this makes finding an exact closed-form solution difficult, because now voting behavior is dependent across individuals \emph{and} issues. In this section, we introduce an approximation of these dependent variables as independent variables which, while imperfect, allows us to see how vote trading changes when all voters have the opportunity to trade with a random partner. 

In Appendix \ref{sec:NE_deriv}, we computed the probability of the trade being the swing vote as $\binom{n-2}{\frac{n-3}{2}} (Q_1^-)^{\frac{n-1}{2}} (Q_1^+)^{\frac{n-3}{2}}$. When allowing all pairs of voters to trade, we replace these $Q$ terms with new quantities that take into account the probability that the vote was traded away. To see the details of this approximation, see Appendix \ref{sec:group_welfare_deriv}. 

Appendix \ref{sec:group_welfare_deriv} also includes recreations of all previous results under this new model. Many of the results are changed only slightly, since it is fruitless to try to take other trades into account unless you can predict the direction of those trades. However, we do see a change in the group welfare implications of trades for the distribution in Figure \ref{fig:beneficial_dist}, an example that was engineered specifically to encourage trades in one direction. In the original model, vote trading was very valuable because it made it more likely that the group would choose the high preference intensity options. Once all voters are allowed to trade, however, the probability of any one trade being influential goes way down, and the welfare gains of trading also decrease.

\section{Discussion}

The study of vote trading is far from over; this model of vote trading makes several assumptions of varying strength. Voters have zero information about the preferences of their trading partners, but complete information about the preferences of the entire group. For simplicity, we restrict voters to one-for-one vote trades, ignoring the possibilities of trading multiple unimportant issues for one supercritical issue or more than two voters gathering to swap a multitude of issues between them all. The primary model also assumes myopic trades where the group votes immediately after a single trade is made, removing the need to consider the impact of other trades on the distribution of utilities. When we loosen this assumption, our approximation reveals changes in the trades being offered, primarily when a distribution is only conducive to trades in one direction. All these simplifications and approximations are avenues for future investigation. For example, equipping voters with partial information about the preferences of their partners could illuminate the middle ground between previous complete information models and the zero information model presented here. 

Despite these limitations, this paper makes several new contributions to our understanding of vote trading. First, it highlights that the welfare implications of vote trading are dependent on the underlying utility distributions. The question of if vote trading improves group welfare does not have a simple yes/no answer. However, for many reasonable utility distributions, including distributions drawn from real data, vote trading adds value for the entire group infrequently, at best. One ray of hope here is that the vote that are beneficial for the whole group are also the most beneficial for the voters themselves, so if voters restrain themselves to only trading on issues where they stand to gain the most utility, the value for the group could also go up.

Second, the toy model presented here gives real insights into the dynamics of vote trading. We see the role of symmetry in the Nash equilibria and we get the counterintuitive result that some utility pairs can be traded in either direction while others cannot be traded for positive value at all. Although these equilibria are not unique in general, our investigation suggests that in most realistic scenarios, rational vote traders should converge to a single Nash equilibrium.

Finally, this paper outlines a probabilistic method for thinking about vote trading, and we believe this could open up new avenues in the study of vote trading.

\section*{Acknowledgements}
We wish to thank Nicholas Christakis, Larry Samuelson, Dimitrios Xefteris, and Alessandra Casella for their helpful comments. This work is supported by the Human Nature Lab and the Sunwater Institute.

\section*{Data Availability}
The data and code used for this project are publicly available at 

https://github.com/MattJonesMath/Vote\_Trading\_Equilibria.

\bibliography{refs}

\newpage

\begin{appendices}

\section{Proof of Theorem \ref{thm:ne}}\label{sec:NE_deriv}

Suppose we have two issues, $t_1$ and $t_2$, and two players, $v_1$ and $v_2$. $v_1$ will be giving away his vote on $t_2$ in exchange for an additional vote on $t_1$ and $v_2$ will be giving away her vote on $t_1$ in exchange for an additional vote on $t_2$. 

Utilities on the two issues are assigned according to a joint probability distribution $f(x,y): [-1,1]^2 \to \mathbb R$ where $x$ is the utility on $t_1$ and $y$ is the utility on $t_2$. By integrating this function, we can determine the probability a random voter supports or opposes $t_1$ and $t_2$:

\begin{equation*}
    \text{Prob}(t_1 \text{ utility} > 0) = Q_1^+ = \int_0^1 \int_{-1}^1 f(x,y) dy dx
\end{equation*}
\begin{equation*}
    \text{Prob}(t_1 \text{ utility} < 0) = Q_1^- = \int_{-1}^0 \int_{-1}^1 f(x,y) dy dx
\end{equation*}
\begin{equation*}
    \text{Prob}(t_2 \text{ utility} > 0) = Q_2^+ = \int_0^1 \int_{-1}^1 f(x,y) dx dy
\end{equation*}
\begin{equation*}
    \text{Prob}(t_2 \text{ utility} < 0) = Q_2^- = \int_{-1}^0 \int_{-1}^1 f(x,y) dx dy
\end{equation*}

We assume all four quantities are nonzero; otherwise, trading fails as certain individuals have no incentive to trade. As shown in Figure \ref{fig:dependent_both_players}, there will be eight regions of interest, $R_i$, representing the eight different types of trades, and each is defined by a constant $\theta_i$ which is the angle of the bounding line with the x or y axis.

To compute the probability that a vote trader supports or opposes an issue, we integrate over these regions and normalize as needed. 

\begin{equation*}
    I_* = \int \int_{R_*} f(x,y) dx dy
\end{equation*}

\begin{equation*}
    I_{S1} = I_1 + I_2 + I_3 + I_4
\end{equation*}

\begin{equation*}
    I_{S2} = I_5 + I_6 + I_7 + I_8
\end{equation*}

With this, we can write down all the necessary probabilities:

\begin{equation*}
    P(u_1>0) = \frac{I_1+I_4}{I_{S1}}
\end{equation*}
\begin{equation*}
    P(u_1<0) = \frac{I_2+I_3}{I_{S1}}
\end{equation*}
\begin{equation*}
    P(u_2>0) = \frac{I_1+I_2}{I_{S1}}
\end{equation*}
\begin{equation*}
    P(u_2<0) = \frac{I_3+I_4}{I_{S1}}
\end{equation*}
\begin{equation*}
    P(u_3>0) = \frac{I_5+I_6}{I_{S2}}
\end{equation*}
\begin{equation*}
    P(u_3<0) = \frac{I_7+I_8}{I_{S2}}
\end{equation*}
\begin{equation*}
    P(u_4>0) = \frac{I_5+I_8}{I_{S2}}
\end{equation*}
\begin{equation*}
    P(u_4<0) = \frac{I_6+I_7}{I_{S2}}
\end{equation*}

Now we demonstrate how to use Equation \eqref{eq:expectation} for $v_1$ when $u_1>0$ and $u_2>0$, which determines the optimal value of $\theta_1$. The other seven cases are similar.

First, consider what happens to the vote on $t_1$. Nothing happens unless $v_2$ has negative utility on $t_1$, i.e. $u_4 < 0$. By above, this occurs with probability $\frac{I_6 + I_7}{I_{S2}}$. Conditional on this being true, the vote only changes if $v_2$ was the swing vote, which means that $\frac{n-1}{2}$ of the non-trading voters have negative utility on $t_1$ and the remaining $\frac{n-3}{2}$ voters have positive utility. This happens with probability $\binom{n-2}{\frac{n-3}{2}} (Q_1^-)^{\frac{n-1}{2}} (Q_1^+)^{\frac{n-3}{2}}$. We multiply these quantities by $2u_1$ to get the expected value for $v_1$ of gaining an extra vote on $t_1$. We repeat the process to compute the expected value of giving up a vote on $t_2$ to get

\begin{align*}
&E(\text{value of trade}) = \\
&2 u_1 \frac{I_6 + I_7}{I_{S2}} \binom{n-2}{\frac{n-3}{2}} (Q_1^-)^{\frac{n-1}{2}} (Q_1^+)^{\frac{n-3}{2}} \\
    -&2 u_2 \frac{I_7 + I_8}{I_{S2}} \binom{n-2}{\frac{n-3}{2}} (Q_2^-)^{\frac{n-3}{2}} (Q_2^+)^{\frac{n-1}{2}} \\
     = & \frac{2}{I_{S2}} \binom{n-2}{\frac{n-3}{2}} \left( u_1(I_6+I_7)(Q_1^-)^\frac{n-1}{2} (Q_1^+)^\frac{n-3}{2} - u_2 (I_7+I_8) (Q_2^-)^\frac{n-3}{2} (Q_2^+)^\frac{n-1}{2}\right)
\end{align*}

By setting the last line of this equality to zero and solving for $\theta_1 = \arctan\left( \left|\frac{u_2}{u_1}\right| \right)$, we find the trades that have expected value zero.

\begin{equation}
    \theta_1 = \arctan \left( \frac{I_6+I_7}{I_7+I_8}\frac{(Q_1^-)^\frac{n-1}{2} (Q_1^+)^\frac{n-3}{2}}{(Q_2^-)^\frac{n-3}{2} (Q_2^+)^\frac{n-1}{2}} \right)
\end{equation}

To ensure that this is well-defined, if $I_7 + I_8 = 0$, then $\theta_1 = \frac{\pi}{2} = \lim_{x \to \infty} \arctan(x)$. We leave this function undefined when the numerator and denominator are zero; in this case, no trade will happen and so there is no change in value, positive or negative.

Repeat this process for the other seven $\theta_i$ values. Solutions to this set of equations represent a Nash equilibrium, since they are a best response to themselves. In fact, it is a strict Nash equilibrium. Any deviation from this strategy, when played against this strategy, either offers trades with negative expected value or refuses trades with positive expected value; in both cases, the deviant strategy has a lower expected payoff. 

\section{Proof of Corollary \ref{thm:existence}}\label{sec:exist_deriv}

Let $(\theta_1, \dots, \theta_8)$ be a solution to Equations~\eqref{eq:k1} - \eqref{eq:k8}. First, we show that if this point is not exactly the origin, then it must not be located near the origin.

\begin{lemma}
    Let $Q_{min} = min \left\{ \frac{(Q_1^-)^\frac{n-1}{2} (Q_1^+)^\frac{n-3}{2}}{(Q_2^-)^\frac{n-3}{2} (Q_2^+)^\frac{n-1}{2}}, \dots, \frac{(Q_2^-)^\frac{n-3}{2} (Q_2^+)^\frac{n-1}{2}}{(Q_1^-)^\frac{n-3}{2} (Q_1^+)^\frac{n-1}{2}} \right\}$ and $\theta_{min} = \arctan(Q_{min})$. 

    If $\theta_i > 0$ for any $i$, then all $\theta_i>0$ for all $i$.
    
    Furthermore, if any $\theta_i > 0$, then at least two of $\theta_1$ through $\theta_4$ must be greater than or equal to $\theta_{min}$, as well as at least two of $\theta_5$ through $\theta_8$.
\end{lemma} 

\begin{proof}
    Suppose without loss of generality that $\theta_1 \neq 0$. 

    Using Equations~\eqref{eq:k1} - \eqref{eq:k8}, $\theta_1 > 0 \implies \theta_7, \theta_8 > 0 \implies \theta_2, \theta_3, \theta_4 > 0 \implies \theta_5, \theta_6 > 0$.
    
    By Equation~\eqref{eq:k7}, either $\frac{I_1+I_2}{I_1+I_4} >= 1$ and therefore $\theta_7 >= \theta_{min}$, or $I_4 > I_2$, in which case $\theta_5 > \theta_{min}$. 

    Likewise, by Equation~\eqref{eq:k8}, $\theta_8 >= \theta_{min}$ or $I_3>I1$ and therefore $\theta_6 > \theta_{min}$. With the same method, we get that $\theta_1 >= \theta_{min}$ or $\theta_6 > \theta_{min}$ and that $\theta_2 >= \theta_{min}$ or $\theta_4 > \theta_{min}$ .
\end{proof}

Note that this implies that if any $\theta_i > 0$, then Equations~\eqref{eq:k1} - \eqref{eq:k8} are all well-defined, since we never have $\frac{0}{0}$.

For each $(\theta_1, \dots, \theta_8)$, there is a best response, i.e. a set of trades that have positive value and a set of trades that have negative value. Equations~\eqref{eq:k1} - \eqref{eq:k8} tell us how to define this function $BR: [0, \frac{\pi}{2}]^8 \to [0, \frac{\pi}{2}]^8$.

\begin{equation}
    (\theta_1, \dots, \theta_8) \mapsto \left( 
    \arctan \left( \frac{I_6(\theta_6)+I_7(\theta_7)}{I_7(\theta_7)+I_8(\theta_8)}\frac{(Q_1^-)^\frac{n-1}{2} (Q_1^+)^\frac{n-3}{2}}{(Q_2^-)^\frac{n-3}{2} (Q_2^+)^\frac{n-1}{2}} \right), \dots, \arctan \left( \frac{I_1(\theta_1)+I_2(\theta_2)}{I_2(\theta_2)+I_3(\theta_3)} \frac{(Q_2^-)^\frac{n-3}{2} (Q_2^+)^\frac{n-1}{2}}{(Q_1^-)^\frac{n-3}{2} (Q_1^+)^\frac{n-1}{2}} \right)
    \right)
\end{equation}

Let $R$ be the subset of $[0, \frac{\pi}{2}]^8$ where $\theta_1 + \theta_3 >= \theta_{min}$, $\theta_2 + \theta_4 >= \theta_{min}$, $\theta_5 + \theta_7 + >= \theta_{min}$, and $\theta_6 + \theta_8 >= \theta_{min}$. $R$ is clearly convex and compact. By the above lemma, $BR$ maps $R$ to $R$, is well-defined, and is continuous. Therefore, we can apply Brouwer's Fixed Point Theorem and are done. We have bounded away from the origin, so we know that the fixed point, which is a Nash equilibrium, is not the trivial Nash equilibrium.

\section{Group Welfare Derivations}\label{sec:group_welfare_deriv}

Here, we derive bounds on the trades that provide benefits (in expectation) for the entire group, not just the vote traders. We use a modification of Equation \eqref{eq:expectation} that includes the utility of all members of the group.

\begin{equation}\label{eq:expectation_group}
E(u_1,u_2) = 2 E(\text{value for all voters}) P(A)P(C|A) \\
- 2 E(\text{value for all voters}) P(B)P(D|B)
\end{equation}

We can write down these terms explicitly for $u_1>0$ and $u_2>0$. The first term in the parentheses is the vote trader, the second term is the trading partner, the third term is the $\frac{n-1}{2}$ voters that don't agree with the trader, and the fourth term is the $\frac{n-3}{2}$ voters that do agree. 

\begin{align*}
   E(u_1,u_2) =& \frac{I_6 + I_7}{I_{S2}} \binom{n-2}{\frac{n-3}{2}} (Q_1^-)^{\frac{n-1}{2}} (Q_1^+)^{\frac{n-3}{2}} \\
   \Bigg(&2u_1 +2 \frac{1}{I_6+I_7} \iint_{R_6 \cup R_7} x f(x,y) dx dy \\
    &+2 \frac{n-1}{2}\frac{1}{Q_1^-} \int_{-1}^0\int_{-1}^1 x f(x,y) dy dx +2 \frac{n-3}{2} \frac{1}{Q_1^+} \int_0^1\int_{-1}^1 x f(x,y) dy dx \Bigg)\\
    + &\frac{I_7 + I_8}{I_{S2}} \binom{n-2}{\frac{n-3}{2}} (Q_2^-)^{\frac{n-3}{2}} (Q_2^+)^{\frac{n-1}{2}} \cdot -1 \cdot \\ \Bigg(&2u_2 +2 \frac{1}{I_7+I_8} \iint_{R_7 \cup R_8} y f(x,y) dx dy \\
    &+2 \frac{n-3}{2} \frac{1}{Q_2^-} \int_{-1}^0\int_{-1}^1 y f(x,y) dx dy +2 \frac{n-1}{2}\frac{1}{Q_2^+} \int_0^1\int_{-1}^1 y f(x,y) dx dy\Bigg)
\end{align*}

Like before, we can set this expression equal to zero and solve to find the trades with expected value zero from the trade. After removing some common terms, we get 

\begin{equation}\label{eq:beneficial_eq_first}
    \text{Quadrant 1:} \quad \frac{a_1}{c_1}u_1 + \frac{a_1 b_1}{c_1} - d_1 \geq u_2
\end{equation}

where 

$$a_1 = (I_6+I_7)(Q_1^-)^\frac{n-1}{2}(Q_1^+)^\frac{n-3}{2}$$

$$b_1 = \frac{1}{I_6+I_7}\iint_{R_6 \cup R_7} x f(x,y) dx dy + \frac{n-1}{2}\frac{1}{Q_1^-} \int_{-1}^0\int_{-1}^1 x f(x,y) dy dx + \frac{n-3}{2} \frac{1}{Q_1^+} \int_0^1\int_{-1}^1 x f(x,y) dy dx$$

$$c_1 = (I_7+I_8)(Q_2^-)^\frac{n-3}{2}(Q_2^+)^\frac{n-1}{2}$$

$$d_1 = \frac{1}{I_7+I_8} \iint_{R_7 \cup R_8} y f(x,y) + \frac{n-3}{2} \frac{1}{Q_2^-} \int_{-1}^0\int_{-1}^1 y f(x,y) dx dy + \frac{n-1}{2}\frac{1}{Q_2^+} \int_0^1\int_{-1}^1 y f(x,y) dx dy$$

Notice that this line is parallel to the boundary of $R_1$, just shifted by a factor of $\frac{a_1 b_1}{c_1} - d_1$. If this term is sufficiently negative, there may be no trades with $u_1>0$ and $u_2 > 0$ that have positive value for the group.

The exact same process can be repeated for all eight types of trades. We give the final expressions now.

\begin{equation}
    \text{Quadrant 2:} \quad -\frac{a_2}{c_2} u_1 - \frac{a_2 b_2}{c_2} - d_2 \geq u_2
\end{equation}

\begin{equation}
    \text{Quadrant 3:} \quad \frac{a_3}{c_3} u_1 + \frac{a_3 b_3}{c_3} - d_3 \leq u_2
\end{equation}

\begin{equation}
    \text{Quadrant 4:} \quad -\frac{a_4}{c_4} u_1 - \frac{a_4 b_4}{c_4} - d_4 \leq u_2
\end{equation}

\begin{equation}
    \text{Quadrant 5:} \quad u_3 \geq \frac{c_5}{a_5} u_4 + \frac{c_5 d_5}{a_5} - b_5
\end{equation}

\begin{equation}
    \text{Quadrant 6:} \quad u_3 \geq - \frac{c_6}{a_6}u_4 - \frac{c_6 d_6}{a_6} - b_6
\end{equation}

\begin{equation}
    \text{Quadrant 7:} \quad u_3 \leq \frac{c_7}{a_7} u_4 + \frac{c_7 d_7}{a_7} - b_7
\end{equation}

\begin{equation}\label{eq:beneficial_eq_last}
    \text{Quadrant 8:} \quad u_3 \leq - \frac{c_8}{a_8} u_4 - \frac{c_8 d_8}{a_8} - b_8
\end{equation}

The coefficients can be calculated as:

$$a_2 = (I_5 + I_8)(Q_1^-)^\frac{n-3}{2}(Q_1^+)^\frac{n-1}{2}$$

$$b_2 = \frac{1}{I_5 + I_8}\iint_{R_5 \cup R_8} x f(x,y) dx dy + \frac{n-3}{2}\frac{1}{Q_1^-} \int_{-1}^0 \int_{-1}^1 x f(x,y) dy dx + \frac{n-1}{2}\frac{1}{Q_1^+} \int_0^1\int_{-1}^1x f(x,y)dy dx$$

$$c_2 = c_1$$

$$d_2 = d_1$$

$$a_3 = a_2$$

$$b_3 = b_2$$

$$c_3 = (I_5+I_6)(Q_2^-)^\frac{n-1}{2}(Q_2^+)^\frac{n-3}{2}$$

$$d_3 = \frac{1}{I_5+I_6} \iint_{R_5 \cup R_6} y f(x,y) dx dy + \frac{n-1}{2} \frac{1}{Q_2^-} \int_{-1}^0 \int_{-1}^1 y f(x,y) dx dy + \frac{n-3}{2} \frac{1}{Q_2^+} \int_0^1 \int_{-1}^1 y f(x,y) dx dy$$

$$a_4 = a_1$$

$$b_4 = b_1$$

$$c_4 = c_3$$

$$d_4 = c_3$$

$$a_5 = (I_3+I_4)(Q_2^-)^\frac{n-1}{2}(Q_2^+)^\frac{n-3}{2}$$

$$b_5 = \frac{1}{I_3+I_4} \iint_{R_3 \cup R_4} y f(x,y)dxdy + \frac{n-1}{2}\frac{1}{Q_2^-} \int_{-1}^0\int_{-1}^1 y f(x,y) dx dy + \frac{n-3}{2}\frac{1}{Q_2^+} \int_0^1\int_{-1}^1 y f(x,y) dx dy$$

$$c_5 = (I_2 + I_3)(Q_1^-)^\frac{n-3}{2}(Q_1^+)^\frac{n-1}{2}$$

$$d_5 = \frac{1}{I_2+I_3} \iint_{R_2 \cup R_3} x f(x,y) dx dy + \frac{n-3}{2}\frac{1}{Q_1^-} \int_{-1}^0\int_{-1}^1 x f(x,y)dydx + \frac{n-1}{2}\frac{1}{Q_1^+} \int_0^1\int_{-1}^1 x f(x,y) dy dx$$

$$a_6 = a_5$$

$$b_6 = b_5$$

$$c_6 = (I_1+I_4)(Q_1^-)^\frac{n-1}{2}(q_1^+)^\frac{n-3}{2}$$

$$d_6 = \frac{1}{I_1+I_4}\iint_{R_1 \cup R_4} x f(x,y) dx dy + \frac{n-1}{2}\frac{1}{Q_1^-} \int_{-1}^0 \int_{-1}^1 x fx,y) dy dx + \frac{n-3}{2}\frac{1}{Q_1^+} \int_0^1 \int_{-1}^1 x f(x,y) dy dx$$

$$a_7 = (I_1+I_2)(Q_2^-)^\frac{n-3}{2}(Q_2^+)^\frac{n-1}{2}$$

$$b_7 = \frac{1}{I_1+I_2}\iint_{R_1 \cup R_2} y f(x,y) dx dy + \frac{n-3}{2} \frac{1}{Q_2^-} \int_{-1}^0\int_{-1}^1 y f(x,y) dx dy + \frac{n-1}{2} \frac{1}{Q_2^+} \int_0^1\int_{-1}^1 y f(x,y) dx dy$$

$$c_7 = c_6$$

$$d_7 = d_6$$

$$a_8 = a_7$$

$$b_8 = b_7$$

$$c_8 = c_5$$

$$d_8 = d_5$$

\section{Group-wide Trading Details}\label{sec:group_trading_details}

We will need to know the probability that a voter is willing to trade in either direction, so we define the following quantities:

\begin{equation}
    J_1 = \iint_{R_1 \cap R_5} f(x,y) dx dy
\end{equation}

\begin{equation}
    J_2 = \iint_{R_2 \cap R_6} f(x,y) dx dy
\end{equation}

\begin{equation}
    J_3 = \iint_{R_3 \cap R_7} f(x,y) dx dy
\end{equation}

\begin{equation}
    J_4 = \iint_{R_4 \cap R_8} f(x,y) dx dy
\end{equation}

A voter $v$ who is deciding what trade to offer needs to know the probability that the trade will change the outcome of the vote. Consider a random other voter $w$ in the population that is not paired with $v$. First, let us determine the probability that $w$ ultimately casts a vote in support of $t_1$.

$w$ initially supports $t_1$ with probability $Q_1^+$. With probability $\frac{n-3}{n-2}$, $w$ has the opportunity to trade. 

There is also a chance that $w$ gives away their vote to someone that opposes $t_1$. This happens if $w$'s utilities fall in $R_5$ or $R_8$ and $w$'s trading partner is in $R_2$ or $R_3$, which happens with probability $(I_5 + I_8)(I_2+I_3)$ which we subtract from the initial probability. However, if $w$ and their partner are both willing to trade in either direction, then $w$ only gives up their vote half the time, so the probability that $w$ has a positive utility on $t_1$ but trades it away is $(I_5 + I_8)(I_2+I_3) - \frac{1}{2} (J_1 + J_4)(J_2+J_3)$.

Similarly, $w$ could also vote for $t_1$ if they initially oppose the issue but give away their vote to someone who supports it. Therefore, we must add $(I_6 + I_7)(I_1+I_4) - \frac{1}{2} (J_2+J_3)(J_1+J_4)$. 

When we add all three of these terms together, we have the probability that voter $w$ votes in favor of $t_1$:

$$\mathcal{Q}_1^+ = Q_1^+ + \frac{n-3}{n-2} \left( (I_6+I_7)(I_1+I_4) - \frac{1}{2} (J_2+J_3)(J_1+J_4) - (I_5+I_8)(I_2+I_3) + \frac{1}{2} (J_1+J_4)(J_2+J_3)  \right)$$

The $J$ terms cancel out, and we are left with

\begin{equation}\label{eq:calQ1plus}
    \mathcal{Q}_1^+ = Q_1^+ + \frac{n-3}{n-2} \bigg( (I_6+I_7)(I_1+I_4) - (I_5+I_8)(I_2+I_3)  \bigg). 
\end{equation}

A similar process gives the other necessary probabilities.

\begin{equation}\label{eq:calQ1minus}
    \mathcal{Q}_1^- = Q_1^- + \frac{n-3}{n-2} \bigg( (I_5+I_8)(I_2+I_3) - (I_6+I_7)(I_1+I_4) \bigg)
\end{equation}

\begin{equation}\label{eq:calQ2plus}
    \mathcal{Q}_2^+ = Q_2^+ + \frac{n-3}{n-2} \bigg( (I_3+I_4)(I_5+I_6) - (I_1+I_2)(I_7+I_8) \bigg)
\end{equation}

\begin{equation}\label{eq:calQ2minus}
    \mathcal{Q}_2^- = Q_2^- + \frac{n-3}{n-2} \bigg( (I_1+I_2)(I_7+I_8) - (I_3+I_4)(I_5+I_6) \bigg)
\end{equation}

Now, when computing the Nash equilibrium using Equations~\eqref{eq:k1} - \eqref{eq:k8}, we replace the $Q$ terms with $\mathcal{Q}$ terms from Equations~\eqref{eq:calQ1plus} - \eqref{eq:calQ2minus} to approximate the probability of being the pivotal vote after all other voter pairs have traded. We can make this same replacement in the $a_i$ and $c_i$ terms when computing group welfare implications. Note that we do not replace the $Q$ terms in the $b_i$ and $d_i$ equations, since those are computing welfare, not the probability of being a pivotal vote.

When $f$ is point-symmetric around the origin, it has the naive Nash equilibrium, all $I_i$s are equal, and almost all terms in Equations~\eqref{eq:calQ1plus} - \eqref{eq:calQ2minus} cancel out. There is no change to the equilibrium between the myopic model and the model that allows all voters to trade votes.

We end by recreating Figures~\ref{fig:pw_const_dist}, \ref{fig:beneficial_dist}, and \ref{fig:anes_fig} with the new equilibria found when all voters are paired up and allowed to trade. The changes are modest, and most pronounced when the distribution $f$ is designed to promote trades in one direction.

\begin{figure}
    \centering
    \includegraphics[width = \textwidth]{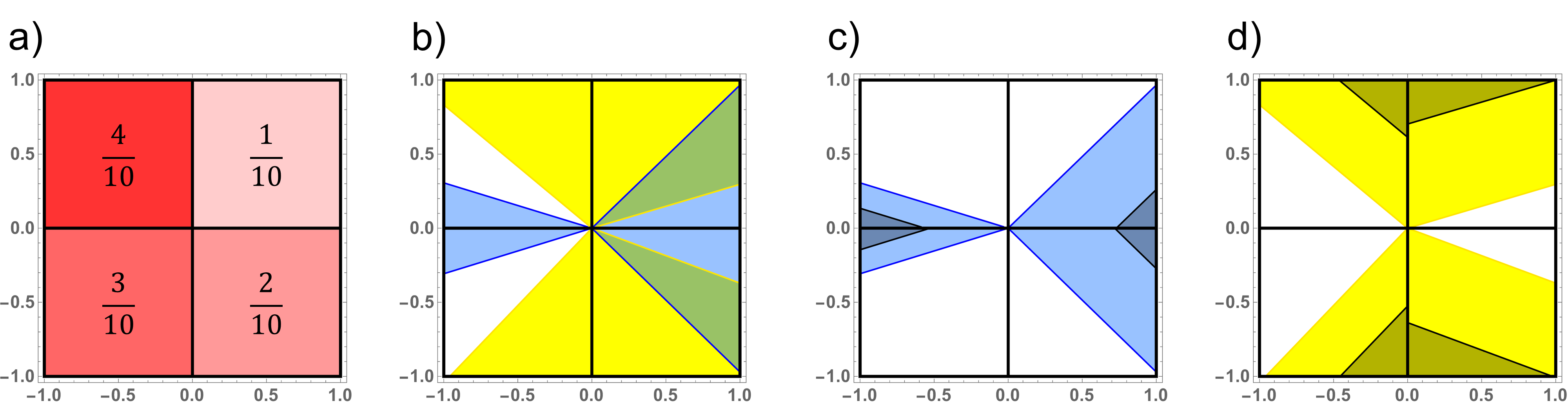}
    \caption{The Nash equilibrium and welfare-improving trades when all voters are paired up and allowed to trade for the distribution in a. This figure can be compared to Figure~\ref{fig:pw_const_dist} to see the small changes that happen when multiple trades are permitted.}
    \label{fig:pw_const_dist_gt}
\end{figure}

\begin{figure}
    \centering
    \includegraphics[width = \textwidth]{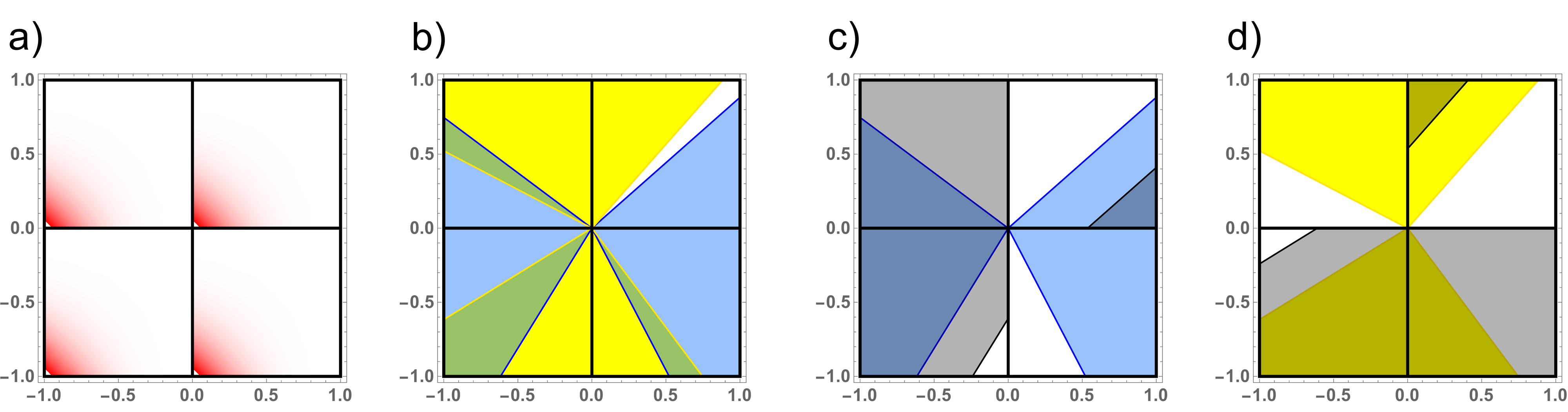}
    \caption{The Nash equilibrium and welfare-improving trades when all voters are paired up and allowed to trade for the distribution in a. This figure can be compared to Figure~\ref{fig:beneficial_dist} to see that when multiple trades are permitted, some trades located in the first quadrant become detrimental to group welfare.}
    \label{fig:beneficial_dist_gt}
\end{figure}

\begin{figure}
    \centering
    \includegraphics[width = \textwidth]{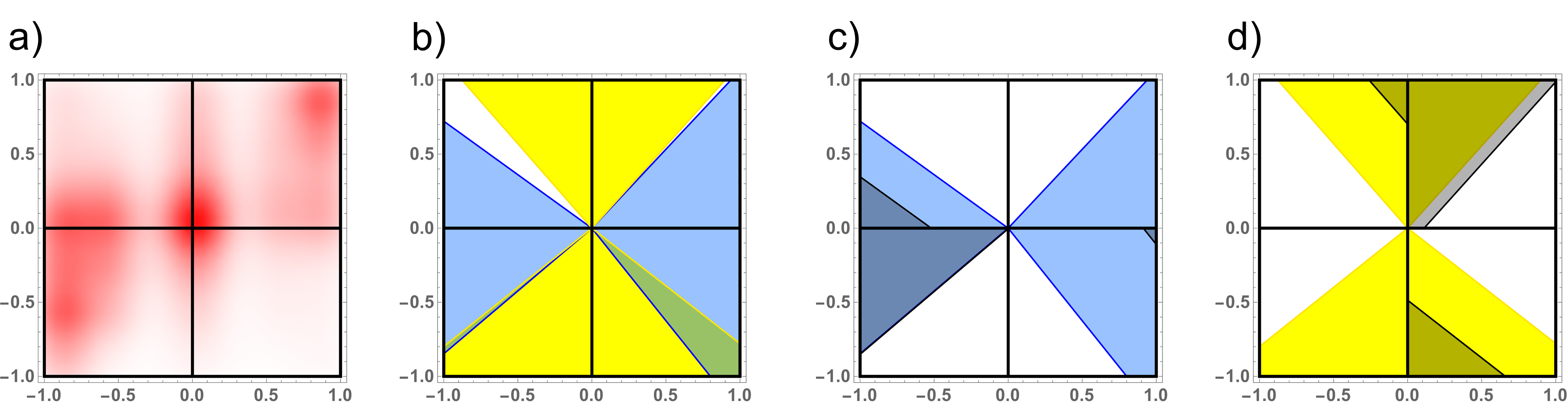}
    \caption{The Nash equilibrium and welfare-improving trades when all voters are paired up and allowed to trade on the issues from the ANES data. This figure can be compared to Figure~\ref{fig:pw_const_dist} to see the small changes that happen when multiple trades are permitted.}
    \label{fig:anes_fig_gt}
\end{figure}

\end{appendices}

\end{document}